%% file: main.tex
\documentclass[submission,copyright,creativecommons]{eptcs}
\providecommand{\event}{Gandalf'22} 
\usepackage{breakurl}             
\usepackage{underscore}           

\usepackage{amsmath}
\usepackage{amssymb}
\usepackage{mathtools}
\usepackage[
    linecolor=blue,
    bordercolor=blue,
    backgroundcolor=white]{todonotes}
\usepackage{environ}
\usepackage{stmaryrd}
\usepackage{subcaption}
\usepackage{array}

\usepackage{graphicx}

\usepackage{cite}

\input{macros}

\usepackage{tikz}
\usetikzlibrary{shapes}
\usetikzlibrary{patterns}
\usetikzlibrary{arrows.meta}
\usetikzlibrary{arrows}
\usetikzlibrary{decorations.pathmorphing}
\tikzstyle{data}=[rectangle split,rectangle split parts=2,draw,text centered, rectangle split part fill={blue!20,blue!20}]
\tikzstyle{dataredblue}=[rectangle split,rectangle split parts=2,draw,text centered, rectangle split part fill={red!30,blue!20}]
\tikzstyle{databluered}=[rectangle split,rectangle split parts=2,draw,text centered, rectangle split part fill={blue!20,red!30}]
\tikzstyle{datablueblue}=[rectangle split,rectangle split parts=2,draw,text centered, rectangle split part fill={blue!20,blue!20}]
\tikzstyle{dataredred}=[rectangle split,rectangle split
parts=2,draw,text centered, rectangle split part fill={red!30,red!30}]
\tikzstyle{datawhite}=[rectangle split,rectangle split parts=2,draw,text centered]
\tikzstyle{dataone}=[rectangle split,rectangle split parts=1,draw,text centered]

\usetikzlibrary{arrows}

\title{On the Existential Fragments of Local First-Order Logics with Data}
\author{Benedikt Bollig
\institute{CNRS, LMF, ENS Paris-Saclay \\Universit{\'e} Paris-Saclay, France}
\and
Arnaud Sangnier
\institute{IRIF, Universit\'e  Paris Cit\'e \\CNRS,France}
\and
Olivier Stietel
\institute{CNRS, LMF, ENS Paris-Saclay \\ Universit{\'e} Paris-Saclay,
  France\\IRIF, Universit\'e  Paris Cit\'e \\CNRS,France}
}
\def\titlerunning{On the Existential Fragments of Local First-Order Logics with Data}
\def\authorrunning{B. Bollig, A. Sangnier  \& O. Stietel}
\begin{document}
\maketitle

\begin{abstract}We study first-order logic over unordered structures
  whose elements carry a finite number of data values from an
infinite domain which can be compared wrt. equality. As the
satisfiability problem for this logic is undecidable in
general, in a previous work, we have introduced a family of local fragments that restrict quantification to neighbourhoods of a
given reference point. We provide here the precise
complexity characterisation of the satisfiability problem for the existential fragments of this local
logic depending on the number of data values carried by each element
and the radius of the considered neighbourhoods.
\end{abstract}

\section{Introduction}

\input{introduction}

\section{Structures and first-order logic}
\label{sec:preliminaries}
\input{structures}

\section{Decidability results}
\label{sec:decidability}
\input{decidability}

\section{Undecidability results}
\label{sec:undecidability}
\input{undecidability}

\bibliographystyle{eptcs}
\bibliography{biblio}
\end{document}

%% file: macros.tex
\newtheorem{theorem}{Theorem}

\newtheorem{lemma}{Lemma}

\newtheorem{example}{Example}

\newtheorem{remark}{Remark}
\newenvironment{proof}{\noindent \textbf{Proof:}}{\hfill $\Box$ \medskip}

\DeclareMathOperator{\N}{\mathbb{N}}
\renewcommand{\phi}{\varphi}
\newcommand{\vp}{\varphi}

\newcommand{\Unary}{\Sigma}

\newcommand{\unary}{\sigma}

\newcommand{\uP}[1]{P_{#1}}

\def \AA{\mathfrak{A}}
\def \BB{\mathfrak{B}}
\def \GG{\mathcal{G}}

\newcommand{\nbd}{D} 

\newcommand{\ifunct}{f_1}
\newcommand{\ofunct}{f_2}
\newcommand{\funi}{\ifunct}
\newcommand{\funo}{\ofunct}
\newcommand{\funct}[1]{f_{#1}}

\newcommand{\nData}[2]{\textup{Data}[{#2},{#1}]}
\newcommand{\Values}[1]{\Valuessub{#1}{A}}
\newcommand{\Valuessub}[2]{\mathit{Val}_{#1}(#2)}

\newcommand{\f}[1]{f_{#1}}

\newcommand{\rel}{\sim}
\newcommand{\relsaa}[5]{#4 \mathrel{_{#1}{\sim^{#3}_{#2}}} #5}
\newcommand{\relsaaord}[3]{{_{#1}{\sim^{#3}_{#2}}}}
\newcommand{\relsaord}[2]{{_{#1}{\sim_{#2}}}}
\newcommand{\rels}[4]{#3 \mathrel{_{#1}{\sim}{_{#2}}} #4}

\newcommand{\distaa}[3]{\mathit{d}^{#3}(#1,#2)}
\newcommand{\vprojr}[3]{#1|_{#2}^{#3}}

\newcommand{\datagraph}[1]{\GG(#1)}
\newcommand{\gaifmanish}[1]{\datagraph{#1}}
\newcommand{\Vertex}[1]{\mathit{V}_{#1}}
\newcommand{\Edge}[1]{\mathit{E}_{#1}}
\newcommand{\Ball}[3]{\mathit{B}_{#1}^{#3}(#2)}

\newcommand{\ndFO}[2]{\textup{dFO}[{#2},{#1}]}
\newcommand{\ndFOr}{\textup{dFO}}
\newcommand{\rndFO}[3]{#3\textup{-Loc-dFO}[#2,#1]}
\newcommand{\rndFOr}[1]{#1\textup{-Loc-dFO}}
\newcommand{\eFO}[3]{\exists\textup{-}{#3}\textup{-Loc-dFO}[{#2},{#1}]}
\newcommand{\eFOr}[1]{\exists\textup{-}{#1}\textup{-Loc-dFO}}
\newcommand{\qfFO}[3]{\textup{qf-}{#3}\textup{-Loc-dFO}[{#2},{#1}]}

\newcommand{\Pform}[2]{{#1}(#2)}

\newcommand{\Var}{\mathcal{V}}
\newcommand{\Intrepl}[2]{{I[#1/#2]}}

\makeatletter
\newcommand{\problemtitle}[1]{\gdef\@problemtitle{#1}}
\newcommand{\probleminput}[1]{\gdef\@probleminput{#1}}
\newcommand{\problemquestion}[1]{\gdef\@problemquestion{#1}}
\NewEnviron{decproblem}{
  \problemtitle{}\probleminput{}\problemquestion{}
  \BODY
  \par\addvspace{0.8\baselineskip}
  \noindent
  \fbox{\centering
  \begin{tabular}{rl}
    \multicolumn{2}{l}{\normalsize \@problemtitle} \\
    \hline
    \rule{0pt}{3ex}
    \normalsize \textbf{Input:} & \normalsize \@probleminput \\[0.5ex]
    \normalsize \textbf{Question:} & \normalsize \@problemquestion
  \end{tabular}}
  \par\addvspace{0.8\baselineskip}
}

\DeclareMathOperator{\donc}{\;\Rightarrow\;}
\DeclareMathOperator{\et}{\wedge}
\DeclareMathOperator{\Et}{\bigwedge}
\DeclareMathOperator{\ou}{\vee}
\DeclareMathOperator{\Ou}{\bigvee}

\makeatletter
\newsavebox{\@brx}
\newcommand{\llangle}[1][]{\savebox{\@brx}{\(\m@th{#1\langle}\)}%
  \mathopen{\copy\@brx\kern-0.5\wd\@brx\usebox{\@brx}}}
\newcommand{\rrangle}[1][]{\savebox{\@brx}{\(\m@th{#1\rangle}\)}%
  \mathclose{\copy\@brx\kern-0.5\wd\@brx\usebox{\@brx}}}
\makeatother

\newcommand{\locformr}[3]{\llangle {#1} \rrangle_{#2}^{ #3}}

\newcommand{\nDataSat}[2]{\ensuremath{\textsc{DataSat}(#1,#2)}}

\newcommand{\im}[1]{\mathit{Im}(#1)}

\newcommand{\phit}[1]{\sem{#1}}

\newcommand{\sem}[1]{\ensuremath{[\![#1]\!]}}
\newcommand{\AAas}{\sem{\AA}_{\tuple{a}}}
\newcommand{\tuple}[1]{\vec{#1}}

\newcommand{\Unaryp}{\Unary'}
\newcommand{\Udeci}[1]{\Gamma_{#1}}
\newcommand{\udd}[3]{a_{#1}[{#2},{#3}]}

\newcommand{\inj}{f_{\textup{new}}}

\newcommand{\fp}{f^p}

\newcommand{\phiBun}[1]{\phi_{#1,\Ball{1}{a_p}{}}^{}}
\newcommand{\phiBdeux}{\phi_{\Ball{2}{a_p}{}}^{}}
\newcommand{\phiBdsu}[1]{\phi_{#1,\Ball{2}{a_p}{}\setminus\Ball{1}{a_p}{}}^{}}
\newcommand{\phiun}{\phi_{j,k,a_p}^{r=1}}
\newcommand{\psiun}{\psi_{j,k,a_p}^{r=1}}
\newcommand{\phideux}{\phi_{j,k,a_p}^{r=2}}
\newcommand{\phitrois}{\phi_{j,k,a_p}^{r>2}}
\newcommand{\psideux}{\psi_{j,k,a_p}^{r>1}}

\newcommand{\phitran}{\phi_{\mathit{tran}}}
\newcommand{\phirefl}{\phi_{\mathit{refl}}}
\newcommand{\phiuniq}{\phi_{\mathit{uniq}}}

\newcommand{\phiwf}{\phi_{\mathit{wf}}}
\newcommand{\dout}{d_{\mathit{out}}}
\newcommand{\T}[1]{\ensuremath{[\![#1]\!]}}
\newcommand{\Tp}[1]{\ensuremath{[\![#1]\!]}_{x_p}}
\newcommand{\Tbis}[1]{\ensuremath{[\![#1]\!]'}}
\newcommand{\Tpbis}[1]{\ensuremath{[\![#1]\!]'}_{x_p}}

\newcommand{\psitran}{\psi_{\mathit{tran}}}
\newcommand{\psirefl}{\psi_{\mathit{refl}}}
\newcommand{\psiwf}{\psi_{\mathit{wf}}}

\newcommand{\comp}[1]{{A\setminus{#1}}}

\newcommand{\uge}{\mathsf{ge}}
\newcommand{\addge}[1]{#1_\uge}
\newcommand{\minusge}[1]{#1_{\setminus\uge}}
\newcommand{\AAge}{\addge{\AA}}
\newcommand{\AAminusge}{\minusge{\AA}}

\usetikzlibrary{shapes}
\usetikzlibrary{patterns}
\usetikzlibrary{arrows.meta}
\usetikzlibrary{arrows}
\usetikzlibrary{decorations.pathmorphing}
\tikzstyle{data}=[rectangle split,rectangle split parts=2,draw,text
centered, rectangle split part fill={blue!20,blue!20}]
\tikzstyle{dataredblue}=[rectangle split,rectangle split parts=2,draw,text centered, rectangle split part fill={red!30,blue!20}]
\tikzstyle{databluered}=[rectangle split,rectangle split parts=2,draw,text centered, rectangle split part fill={blue!20,red!30}]
\tikzstyle{datablueblue}=[rectangle split,rectangle split parts=2,draw,text centered, rectangle split part fill={blue!20,blue!20}]
\tikzstyle{dataredred}=[rectangle split,rectangle split
parts=2,draw,text centered, rectangle split part fill={red!30,red!30}]
\tikzstyle{datawhite}=[rectangle split,rectangle split parts=2,draw,text centered]

\usetikzlibrary{arrows}

%% file: introduction.tex
First-order data logic has emerged to specify properties involving infinite data domains. Potential applications include XML reasoning and the specification of concurrent systems and distributed algorithms. The idea is to extend classic mathematical structures by a mapping that associates with every element of the universe a value from an infinite domain. When comparing data values only for equality, this view is equivalent to extending the underlying signature by a binary relation symbol whose interpretation is restricted to an equivalence relation.

Data logics over word and tree structures were studied in \cite{BojanczykMSS09,BojanczykDMSS11}. In particular, the authors showed that two-variable first-order logic on words has a decidable satisfiability problem. Other types of data logics allow \emph{two} data values to be associated with an element \cite{Kieronski05,KieronskiT09}, though they do not assume a linearly ordered or tree-like universe. Again, satisfiability turned out to be decidable for the two-variable fragment of first-order logic. Other notable extensions, either to multiple data values or to totally ordered data domains, include \cite{KaraSZ10,DeckerHLT14,ManuelZ13,Tan14}.

When considering an arbitrary number of first-order variables, which we do in this paper, the decidability frontier is quickly crossed without further constraints as soon as the  number of allowed data in gretar then two \cite{Janiczak-Undecidability-fm53}. One of the restrictions we consider here is locality, an essential concept in first-order logic. It is well known that first-order logic is only able to express local properties: a first-order formula can always be written as a combination of properties of elements that have limited, i.e., bounded by a given radius, distance from some reference points \cite{Han65,Gai82}. In the presence of (several) data values, imposing a corresponding locality restriction on a logic can help ensuring decidability of its satisfiability problem.

In previous work, we considered a local fragment of first-order data logic over structures whose elements (i) are unordered (as opposed to, e.g., words or trees), and (ii) each carries two data values. We showed that the fragment has a decidable satisfiability problem when restricting local properties to radius~1, while it is undecidable for any radius greater than 1.

In the present paper, we study orthogonal local fragments where
global quantification is restricted to being existential
(while quantification inside a local property is still unrestricted).
We obtain decidability for (i) radius 1 and an arbitrary number of data values, and
for (ii) radius 2 and two data values. In all cases, we provide tight complexity upper and lower bounds.
Moreover, these results mark the exact decidability frontier: satisfiability is undecidable as soon as we consider radius 3 in presence of two data values, or radius 2 together with three data values.

To give a possible application domain of our logic, consider distributed algorithms running on a cloud of processes. Those algorithms are usually designed to be correct independently of the number of processes executing them. Every process gets some inputs and produces some outputs, usually from an infinite domain. These may include process identifiers, nonces, etc. Inputs and outputs together determine the behavior of a distributed algorithm. A simple example is leader election, where every process gets a unique id, whereas the output should be the id of the elected leader and so be the same for all processes. To formalize correctness properties and to define the intended input-output relation, it is hence essential to have suitable data logics at hand.

\paragraph{Outline.}
The paper is structured as follows.
In Section~\ref{sec:preliminaries}, we recall important notions such as structures and first-order logic, and we introduce the local fragments considered in this paper. Section~\ref{sec:decidability} presents the decidable cases, whereas, in Section~\ref{sec:undecidability}, we show that all remaining cases lead to undecidability.

\medskip

\noindent
This work was partly supported by the project ANR FREDDA (ANR-17-CE40-0013).

%% file: structures.tex
\subsection{Data Structures}

We define here the class of models we are interested in. It consists of sets of nodes containing data values with the assumption that each node is labeled by a set of predicates and carries the same number of values. We consider hence $\Unary$ a finite set of unary relation symbols (sometimes
called unary predicates) and an integer $\nbd \geq 0$. A \emph{$\nbd$-data structure} over $\Unary$ is a tuple 
$\AA=(A,(P_{\unary})_{\unary \in \Unary},\f{1},\ldots,\f{\nbd})$
(in the following, we simply write $(A,(P_{\unary}),\f{1},\ldots,\f{\nbd})$)
where $A$ is a nonempty finite set,
$P_\unary \subseteq A$ for all $\unary \in \Unary$, and
$\f{i}$s are mappings $A \to \N$.
Intuitively $A$ represents the set of nodes and $f_i(a)$ is the $i$-th data value carried by $a$ for each node $a \in A$.
For $X\subseteq A$, we let $\Valuessub{\AA}{X} = \{\f{i}(a) \mid a \in X, i\in\{1,\ldots,\nbd\}\}$.
The set of all $\nbd$-data structures over $\Unary$
is denoted by $\nData{\nbd}{\Unary}$.

While this representation is often very convenient to represent
data values, a more standard way of
representing mathematical structures is in terms of binary
relations.
For every $(i,j) \in \{1,\ldots,\nbd\} \times \{1,\ldots,\nbd\}$, the mappings
$\f{1},\ldots,\f{\nbd}$ determine a binary relation
${\relsaaord{i}{j}{\AA}} \subseteq A \times A$
as follows:
$\relsaa{i}{j}{\AA}{a}{b}$ iff $\funct{i}(a) = \funct{j}(b)$.
We may omit the superscript $\AA$ if it is clear from the context
and if $\nbd=1$, as there will be only one relation, we way may write $\rel$ for $\relsaord{1}{1}$.

\subsection{First-Order Logic}

Let $\Var = \{x,y,\ldots\}$ be a
countably infinite set of variables. The set $\ndFO{\nbd}{\Unary}$ of first-order formulas interpreted over $\nbd$-data structures
over $\Unary$ is inductively given by the grammar
$\vp ::= \Pform{\unary}{x} \mid \rels{i}{j}{x}{y} \mid x=y \mid \vp \vee \vp \mid \neg \vp \mid \exists x.\vp$,
where $x$ and $y$ range over $\Var$, $\unary$ ranges over $\Unary$, and $i,j \in \{1,\ldots,\nbd\}$.
We use standard abbreviations such as $\wedge$ for conjunction and
$\Rightarrow$ for implication.
We write $\phi(x_1,\ldots,x_k)$ to indicate that the free variables of $\phi$ are among
$x_1,\ldots,x_k$. We call $\phi$ a \emph{sentence} if it does not contain free variables.

For $\AA=(A,(P_{\unary}),\f{1},\ldots,\f{\nbd}) \in \nData{\nbd}{\Unary}$ and a formula $\phi\in\ndFO{\nbd}{\Unary}$,
the satisfaction relation $\AA \models_I \phi$ is defined wrt.\
an interpretation function $I: \Var \to A$. The
purpose of $I$ is to assign an interpretation to every (free)
variable of $\phi$ so that $\phi$ can be assigned a truth value.
For $x \in \Var$ and $a \in A$, the interpretation function $\Intrepl{x}{a}$
maps $x$ to $a$ and coincides
with $I$ on all other variables.
We then define:
\begin{center}
\begin{tabular}{ll}
$\AA \models_I \Pform{\unary}{x}$ if $I(x) \in P_{\unary}$ &
$\AA \models_I \phi_1 \vee \phi_2$ if $\AA \models_I \phi_1$ or $\AA \models_I \phi_2$\\
$\AA \models_I \rels{i}{j}{x}{y}$ if $\relsaa{i}{j}{\AA}{I(x)}{I(y)}$~~~ &
$\AA \models_I \neg \phi$ if $\AA \not\models_I \phi$\\
$\AA \models_I x = y$ if $I(x) = I(y)$ &
$\AA \models_I \exists x.\phi$ if there is $a \in A$ s.t. $\AA \models_{\Intrepl{x}{a}} \phi$
\end{tabular}
\end{center}

Finally, for a data structure $\AA=(A,(P_{\unary}),\f{1},\ldots,\f{\nbd})$,  a formula $\phi(x_1,\ldots,x_k)$ and $a_1,\ldots,a_k\in A$,
we write $\AA\models\phi(a_1,\ldots a_k)$ if there exists an interpretation function $I$ such that $\AA\models_{I[x_1/a_1]\ldots[x_k/a_k]} \phi$.  In particular, for a sentence $\phi$, we write $\AA\models\phi$ if there exists an interpretation function $I$ such that $\AA \models_I  \phi$.


\begin{example}\label{ex:leader-election}
Assume a unary predicate $\mathrm{leader} \in \Unary$.
The following formula from $\ndFO{2}{\Unary}$ expresses
correctness of a leader-election algorithm: (i)~there is a unique
process that has been elected leader, and (ii)~all processes agree,
in terms of their output values (their second data), on
the identity (the first data) of the leader: 
$ \exists x. (\mathrm{leader}(x) \et \forall y. \big(\mathrm{leader}(y)
\Rightarrow y=x)\big) \et \forall y. \exists x. (\mathrm{leader}(x) \et \rels{1}{2}{x}{y})$.
\end{example}

We are interested here in the satisfiability problem for these logics.
Let $\mathcal{F}$ denote a generic class of first-order formulas, parameterized
by $\Unary$ and $\nbd$. 
In particular, for $\mathcal{F} = \ndFOr$,
we have that $\mathcal{F}[\Unary,\nbd]$ is the class $\ndFO{\nbd}{\Unary}$. The satisfiability problem for $\mathcal{F}$ wrt.\ $\nbd$-data structures
is defined as follows:

\begin{center}
\begin{decproblem}
	\problemtitle{\nDataSat{\mathcal{F}}{\nbd}}
  \probleminput{A finite set $\Unary$ and a sentence $\vp \in \mathcal{F}[\Unary,\nbd]$.}
	\problemquestion{Is there $\AA \in \nData{\nbd}{\Unary}$ such that $\AA \models \vp$\,?}
\end{decproblem}
\end{center}

The following negative result (see \cite[Theorem~1]{Janiczak-Undecidability-fm53}) calls for restrictions of the general logic.
\begin{theorem}\cite{Janiczak-Undecidability-fm53}\label{thm:undec-general}
	The problem $\nDataSat{\ndFOr}{2}$ is undecidable, even when we require that $\Unary = \emptyset$ and we do not use $\relsaord{1}{2}$ and $\relsaord{2}{1}$ in the considered formulas.
\end{theorem}

\subsection{Local First-Order Logic and its existential fragment}



We are interested in logics combining the advantages of  $\ndFO{\nbd}{\Unary}$,
while preserving decidability. With this in mind, we have introduced in \cite{bollig-local-fsttcs21}, for the case of two data values, a \emph{local} restriction, where the scope of quantification in the presence of free variables is restricted to the view of a given element. We present now the defintion of such restrictions adapted to the case of many data values.


First, the view of a node $a$ includes all elements whose distance to $a$ is bounded by a given radius.
It is formalized using the notion of a Gaifman graph (for an
introduction, see~\cite{Libkin04}).
We use here a variant that is
suitable for our setting and that we call \emph{data graph}. 
Given a data structure  $\AA=(A,(P_{\unary}),\f{1},\ldots,\f{\nbd}) \in \nData{\nbd}{\Unary}$, we define its \emph{data graph} $\gaifmanish{\AA}=(\Vertex{\gaifmanish{\AA}},\Edge{\gaifmanish{\AA}})$ with set of vertices $\Vertex{\gaifmanish{\AA}} = A \times\{1,\ldots,\nbd\}$ and set of edges
$\Edge{\gaifmanish{\AA}} = \{ ((a,i),(b,j)) \in \Vertex{\gaifmanish{\AA}} \times \Vertex{\gaifmanish{\AA}} \mid a=b$  or $\rels{i}{j}{a}{b} \}$.
Figure \ref{fig:gaifman-a} provides an example of the graph
$\gaifmanish{\AA}$ for a data structure with $2$ data values.


\newcommand{\selfconnectionright}[1]{\draw[<->, line width=0.7pt] (#1.one east) .. controls +(.4,0) and +(.4,0) .. (#1.two east);}
\newcommand{\selfconnectionleft}[1]{\draw[<->, line width=0.7pt] (#1.one west) .. controls +(-.4,0) and +(-.4,0) .. (#1.two west);}

\begin{figure*}[t]
\centering
	\begin{subfigure}[b]{0.45\textwidth}
\begin{tikzpicture}[node distance=2cm]
	\node [data, label=below left:$a$]                (A)    {1
      \nodepart{two} 2 };
	\node [data, above left of=A,xshift=-1em,label=below:$b$]    (B)    {1 \nodepart{second} 3};
	\node [data, above right of=A,xshift=1em,label=below right:$c$]    (C)    {3 \nodepart{second} 2};
	\node [dataredred, below left of=A,label=below:$d$]    (D)    {5 \nodepart{second} 6};
	\node [dataredred, above right of=B,xshift=1em, label=below:$e$]
    (E)    {4 \nodepart{second} 3};
    \node [data, below right of=A, label=below:$f$]    (F)    {2 \nodepart{second} 7};

	\draw[line width=0.7pt,<->] (A.one north) .. controls +(0,.5) and
    +(.5,0).. (B.one east);
	\draw[line width=0.7pt,<->] (B.two east) .. controls +(2,-0.5) and
    +(-2,.5).. (C.one west);
    \draw[line width=0.7pt,<->] (E.two east) .. controls +(0,0) and
    +(0,0.5).. (C.one north);
    \draw[line width=0.7pt,<->] (E.south west) .. controls +(0,0) and
    +(0.5,.2).. (B.two east);
	\draw[line width=0.7pt,<->] (A.south east) .. controls +(1,-.5)
    and +(0,0).. (C.south west);
    \draw[line width=0.7pt,<->] (A.south) .. controls +(0,-0.5) and
    +(0,0).. (F.one west);
    \draw[line width=0.7pt,<->] (F.north) .. controls +(0,0) and
    +(0,0).. (C.south);
	\selfconnectionright{A};
	\selfconnectionleft{B};
	\selfconnectionright{C};
	\selfconnectionleft{D};
	\selfconnectionleft{E};
	\selfconnectionright{F};

  \end{tikzpicture}
\caption{A data structure $\AA$ and  $\gaifmanish{\AA}$.}
\label{fig:gaifman-a}
	\end{subfigure}
	\unskip\ \vrule\ \hspace{2em}
	\begin{subfigure}[b]{0.45\textwidth}
\begin{tikzpicture}[node distance=2cm]
		\node [data, label=below left:$a$]                (A)    {1
      \nodepart{two} 2 };
	\node [data, above left of=A,xshift=-1em,label=below:$b$]    (B)    {1 \nodepart{second} 3};
	\node [data, above right of=A,xshift=1em,label=below right:$c$]    (C)    {3 \nodepart{second} 2};
	\node [dataredred, below left of=A,label=below:$d$]    (D)    {10 \nodepart{second} 11};
	\node [dataredred, above right of=B,xshift=1em, label=below:$e$]
    (E)    {8 \nodepart{second} 9};
    \node [data, below right of=A, label=below:$f$]    (F)    {2 \nodepart{second} 7};

  \end{tikzpicture}
\caption{$\vprojr{\AA}{a}{2}$: the $2$ view of $a$}
\label{fig:gaifman-b}
\end{subfigure}
\caption{
\label{fig:gaifman}}
\end{figure*}

\medskip

We then define the distance $\distaa{(a,i)}{(b,j)}{\AA} \in \N \cup \{\infty\}$ between two elements $(a,i)$ and $(b,j)$ from $A \times\{1,\ldots,\nbd\}$ as the length of the shortest path from $(a,i)$ to $(b,j)$ in $\gaifmanish{\AA}$. 
For $a \in A$ and $r \in \N$, the \emph{radius-$r$-ball around} $a$ is
the set $\Ball{r}{a}{\AA} = \{ (b,j)\in\Vertex{\gaifmanish{\AA}} \mid
\distaa{(a,i)}{(b,j)}{\AA}\leq r $ for some $i \in
\{1,\ldots,\nbd\}\}$. This ball contains the elements of $\Vertex{\gaifmanish{\AA}}$ that can be reached from $(a,1),\ldots,(a,\nbd)$ through a path of length at most $r$.
On Figure~\ref{fig:gaifman-a} the blue nodes represent $\Ball{2}{a}{\AA}$.

We now define the $r$-view of an element $a$ in the $D$-data structure
$\AA$. Intuitively it is a $D$-data structure with the same elements as $\AA$
but where the data values which are not in the radius-$r$-ball around $a$ are
changed with new values all different one from each other. 
Let $\inj: A \times \{1,\ldots,\nbd\} \to \N \setminus \Values{\AA}$
be an injective mapping. The \emph{$r$-view of $a$ in $\AA$} is the
structure $\vprojr{\AA}{a}{r} = (A,(P_{\unary}),\f{1}',\ldots,\f{n}')
\in \nData{\nbd}{\Unary}$ where its universe is the same as the one of $\AA$ and the unary predicates
stay the same and $\funct{i}'(b)= \funct{i}(b)$ if $(b,i)
\in\Ball{r}{a}{\AA}$, and $\funct{i}'(b)= \inj((b,i))$
otherwise. 
On Figure~\ref{fig:gaifman-b},
the structure $\vprojr{\AA}{a}{2}$ is depicted where the values of the red nodes, not belonging to  $\Ball{2}{a}{\AA}$  have been replaced by
fresh values not in $\{1,\ldots,7\}$.


We are now ready to present the logic $\rndFO{\nbd}{\Unary}{r}$, where $r \in \N$, interpreted over structures from
$\nData{\nbd}{\Unary}$. It is given by the grammar
\begin{align*}
	\vp ~&::=~ \locformr{\psi}{x}{r} \;\mid\; x=y \;\mid\; \exists x.\vp \;\mid\;  \vp \vee \vp \;\mid\;  \neg \vp
\end{align*}
where $\psi$ is a formula from $\ndFO{\nbd}{\Unary}$
with (at most) one free variable $x$. This logic uses the \emph{local
modality} $\locformr{\psi}{x}{r}$ to specify that the formula $\psi$
should be interpreted over the $r$-view of the element associated to
the variable $x$.
For $\AA \in \nData{\nbd}{\Unary}$ and an interpretation function $I$,
we have indeed
$\AA \models_I \locformr{\psi}{x}{r}$ iff $\vprojr{\AA}{I(x)}{r} \models_I \psi$.

\begin{example}
We now illustrate what can be specified by formulas in
the logic $\rndFO{2}{\Unary}{1}$. We can rewrite the formula from Example~\ref{ex:leader-election} so that
it falls into our fragment as follows:
$\exists x. (\locformr{\mathrm{leader}(x)}{x}{1} \et \forall
y. \linebreak[0](\locformr{\mathrm{leader}(y)}{y}{1} \Rightarrow x=y))
\et \forall
y. \linebreak[0] \locformr{\exists x. \mathrm{leader}(x) \et \rels{2}{1}{y}{x}
}{y}{1} $.
The next formula specifies an algorithm in which all processes suggest a value and then choose a new value among those that have been suggested at least twice:
  $\forall x.\locformr{\exists
    y.\exists z. y \neq z \et \rels{2}{1}{x}{y} \et \rels{2}{1}{x}{z} }{x}{1} $. We can also specify partial renaming, i.e., two output values agree only if their input values are the same:
  $\forall x.\locformr{\forall y.(\rels{2}{2}{x}{y}\donc\rels{1}{1}{x}{y}}{x}{1}$.
Conversely, the formula $\forall x.\locformr{\forall y.(\rels{1}{1}{x}{y}\donc\rels{2}{2}{x}{y}}{x}{1}$ specifies partial fusion of equivalences classes.
\end{example}

In \cite{bollig-local-fsttcs21}, we have studied the decidability
status of   the satisfiability problem for $\rndFO{2}{\Unary}{r}$ with $r \geq 1$ and we have shown
that \nDataSat{$\rndFOr{2}$}{2} is undecidable and that
\nDataSat{$\rndFOr{1}$}{2} is decidable when restricting the
formulas (and the view of elements) to binary relations belonging to
the set $\{\rels{1}{1}{}{},\rels{2}{2}{}{},\rels{1}{2}{}{}\}$. Whether
\nDataSat{$\rndFOr{1}$}{2} in its full generality is decidable or not
remains an open problem.

We wish to study here the existential fragment of
$\rndFO{\nbd}{\Unary}{r}$ (with $r \geq 1$ and $D \geq 1$) and
establish when its satisfiability problem is 
decidable. This fragment, denoted by $\eFO{\nbd}{\Unary}{r}$, is given by the grammar
\[\phi ~::=~ \locformr{\psi}{x}{r} \;\mid\; x=y \;\mid\; \neg(x=y) \;\mid\; \exists x.\phi \;\mid\; \phi\ou\phi  \;\mid\; \phi\et\phi \]
where $\psi$ is a formula from $\ndFO{\nbd}{\Unary}$
with (at most) one free variable $x$. The quantifier free fragment $\qfFO{\nbd}{\Unary}{r}$ is
defined by the grammar 
$\phi ~::=~ \locformr{\psi}{x}{r} \;\mid\; x=y \;\mid\; \neg(x=y)
\;\mid\; \phi\ou\phi  \;\mid\; \phi\et\phi $.

\begin{remark}
  Note that for both these
fragments, we do not impose any restrictions on the use of quantifiers in
the formula $\psi$ located under the local modality
$\locformr{\psi}{x}{r}$.
\end{remark}

%% file: decidability.tex
We show here decidability of  $\nDataSat{\eFOr{2}}{2}$ and,
for all $\nbd \geq 0$,
$\nDataSat{\eFOr{1}}{\nbd}$.

\subsection{Preliminary results: 0 and 1 data values}

We introduce two preliminary results we shall use in this
section to obtain new decidability results. First, note that
formulas in $\ndFO{0}{\Unary}$ (i.e. where no data is considered)
correspond to first order logic formulas with a set of predicates and
equality test as a unique relation. As mentioned in  Chapter 6.2.1 of
\cite{borger-classical-springer97}, these formulas belong to the
\emph{L\"owenheim class with equality} also called as the relational
monadic formulas, and their satisfiability problem is in
\textsc{NEXP}. Furthermore, thanks to \cite{etessami-first-ic02} (Theorem 11), we know  that this latter problem is \textsc{NEXP}-hard even
if one considers formulas which use only two variables.

\begin{theorem}\label{thm:0fo}
 $\nDataSat{\ndFOr}{0}$ is \textsc{NEXP}-complete.
\end{theorem}

In \cite{Mundhenk09}, the authors study the satisfiability problem for
Hybrid logic over Kripke structures where the transition relation is
an equivalence relation, and they show that it is
\textsc{N2EXP}-complete. Furthermore in  \cite{Fitting12}, it is shown
that Hybrid logic can be translated to first-order logic in
polynomial time and this holds as well for  the converse
translation. Since $1$-data structures can be interpreted as Kripke
structures with one equivalence relation, altogether this allows us to
obtain the following preliminary result about the satisfiability
problem of $\ndFO{1}{\Unary}$.


\begin{theorem}\label{thm:1fo}
 $\nDataSat{\ndFOr}{1}$ is \textsc{N2EXP}-complete.
\end{theorem}

\subsection{Two data values and balls of radius 2}

In this section, we prove that the satisfiability problem for the
existential fragment of local first-order logic with two data values and balls of radius two is decidable.
To obtain this result we provide a reduction to the satisfiability
problem for first-order logic over $1$-data structures. Our reduction is based on the following intuition. Consider a
$2$-data structure $\AA=(A,(P_{\unary}),\f{1},\f{2}) \in
\nData{2}{\Unary}$ and an element $a \in A$. If we take an
element $b$ in $\Ball{2}{a}{\AA}$, the radius-2-ball around $a$, we
know that either $\f{1}(b)$ or $\f{2}(b)$ is a common value with
$a$. In fact, if $b$ is at distance $1$ of $a$, this holds by definition and 
if $b$ is
at distance $2$ then $b$ shares an element with $c$ at distance $1$ of
$a$ and this element has to be shared with $a$ as well so $b$ ends to
be at distance $1$ of $a$. The
trick consists then in using extra-labels for elements sharing a value with
$a$ that can be forgotten and to keep only the value of $b$ not
present in $a$, this construction leading to a $1$-data structure. It
remains to show that we can ensure that a $1$-data structure is the
fruit of this construction in a formula of $\ndFO{1}{\Unary'}$ (where
$\Unary'$ is obtained from $\Unary$ by adding extra predicates).\\

The first step for our reduction consists in providing a
characterisation for the elements located in the radius-1-ball and the radius-2-ball around
another element.

\begin{lemma}\label{lem:shape-balls}
	Let $\AA=(A,(P_{\unary}),\f{1},\f{2}) \in
    \nData{2}{\Unary}$  and $a,b\in A$ and $j \in \{1,2\}$.  We have:
\begin{enumerate}
\item  $(b,j)\in \Ball{1}{a}{\AA}$ iff there is $i\in \{1,2\}$ such that $\relsaa{i}{j}{\AA}{a}{b}$.
\item $(b,j)\in \Ball{2}{a}{\AA}$ iff there exists $i,k\in \{1,2\}$ such that $\relsaa{i}{k}{\AA}{a}{b}$.
\end{enumerate}
\end{lemma}

\begin{proof} We show both statements:
\begin{enumerate}
\item Since $(b,j)\in \Ball{1}{a}{\AA}$, by definition we have either $b=a$ and in that case $\relsaa{j}{j}{\AA}{a}{b}$ holds, or $b \neq a$ and necessarily there exists $i\in \{1,2\}$ such that $\relsaa{i}{j}{\AA}{a}{b}$.
 \item First, if there exists $i,k\in \{1,2\}$ such that
$\relsaa{i}{k}{\AA}{a}{b}$, then $(b,k)\in \Ball{1}{a}{\AA}$ and $(b,j)\in \Ball{2}{a}{\AA}$ by definition. Assume now that $(b,j)\in \Ball{2}{a}{\AA}$. Hence  there exists $i\in \{1,2\}$ such that $\distaa{(a,i)}{(b,j)}{\AA}\leq 2$.
We perform a case analysis on the value of
$\distaa{(a,i)}{(b,j)}{\AA}$.
\begin{itemize}
  \item \textbf{Case $\distaa{(a,i)}{(b,j)}{\AA}=0$}. In that case
    $a=b$ and $i=j$ and we have  $\relsaa{i}{i}{\AA}{a}{b}$.
   \item \textbf{Case $\distaa{(a,i)}{(b,j)}{\AA}=1$}. In that case,
     $((a,i),(b,j))$ is an edge in the data graph $\gaifmanish{\AA}$
     of $\AA$ which means that $\relsaa{i}{j}{\AA}{a}{b}$ holds.
   \item \textbf{Case $\distaa{(a,i)}{(b,j)}{\AA}=2$}. Note that
     we have by definition $a \neq b$. Furthermore, in that case, there is
     $(c,k)\in A\times\{1,2\}$ such that $((a,i),(c,k))$ and
     $((c,k),(b,j))$ are edges in $\gaifmanish{\AA}$. If $c\neq a$ and
     $c\neq b$, this implies that $\relsaa{i}{k}{\AA}{a}{c}$ and
     $\relsaa{k}{j}{\AA}{c}{b}$, so $\relsaa{i}{j}{\AA}{a}{b}$ and
     $\distaa{(a,i)}{(b,j)}{\AA}=1$ which is a contradiction.
	If $c=a$ and $c\neq b$, this implies that $\relsaa{k}{j}{\AA}{a}{b}$.
	If $c\neq a$ and $c = b$, this implies that $\relsaa{i}{k}{\AA}{a}{b}$.
\end{itemize}
\end{enumerate}
\end{proof}

We consider a formula $\phi=\exists x_1\ldots\exists
x_n.\phi_{qf}(x_1,\ldots,x_n)$ of $\eFO{2}{\Unary}{2}$ in prenex normal form, i.e., such that $\phi_{qf}(x_1,\ldots,x_n)\in\qfFO{2}{\Unary}{2}$. We know that there is a structure $\AA=(A,(P_{\unary})_{\unary \in \Unary},\linebreak[0]\f{1},\f{2})$ in $\nData{2}{\Unary}$ such that $\AA\models\phi$ if and only if there are $a_1,\ldots,a_n \in A $ such that $\AA\models\phi_{qf}(a_1,\ldots,a_n)$.

Let $\AA=(A,(P_{\unary})_{\unary \in \Unary},\f{1},\f{2})$ be a structure  in $\nData{2}{\Unary}$ and a tuple $\tuple{a} = (a_1,\ldots,a_n)$ of elements in $A^n$. We shall present the construction of a $1$-data structure
$\AAas$ in $\nData{1}{\Unaryp}$ (with $\Unary \subseteq \Unaryp$) with the same set of nodes as $\AA$, but where each node carries a single data value. In order to retrieve the data relations that hold in $\AA$ while reasoning over $\AAas$, we introduce extra-predicates in $\Unaryp$ to establish whether a node shares a common value with one of the nodes among $a_1,\ldots,a_n$ in $\AA$.

\begin{figure*}[htbp]
\centering
	\begin{subfigure}[b]{0.45\textwidth}
\begin{tikzpicture}[node distance=2cm]
	\node [data, label=below left:$a$]                (A)    {1
      \nodepart{two} 2 };
	\node [data, above left of=A,xshift=-1em,label=below:$b$]    (B)    {1 \nodepart{second} 3};
	\node [data, above right of=A,xshift=1em,label=below right:$c$]    (C)    {3 \nodepart{second} 2};
	\node [dataredred, below left of=A,label=below:$d$]    (D)    {5 \nodepart{second} 6};
	\node [dataredred, above right of=B,xshift=1em, label=below:$e$]
    (E)    {4 \nodepart{second} 3};
    \node [data, below right of=A, label=below:$f$]    (F)    {2 \nodepart{second} 7};

	\draw[line width=0.7pt,<->] (A.one north) .. controls +(0,.5) and
    +(.5,0).. (B.one east);
	\draw[line width=0.7pt,<->] (B.two east) .. controls +(2,-0.5) and
    +(-2,.5).. (C.one west);
    \draw[line width=0.7pt,<->] (E.two east) .. controls +(0,0) and
    +(0,0.5).. (C.one north);
    \draw[line width=0.7pt,<->] (E.south west) .. controls +(0,0) and
    +(0.5,.2).. (B.two east);
	\draw[line width=0.7pt,<->] (A.south east) .. controls +(1,-.5)
    and +(0,0).. (C.south west);
    \draw[line width=0.7pt,<->] (A.south) .. controls +(0,-0.5) and
    +(0,0).. (F.one west);
    \draw[line width=0.7pt,<->] (F.north) .. controls +(0,0) and
    +(0,0).. (C.south);
	\selfconnectionright{A};
	\selfconnectionleft{B};
	\selfconnectionright{C};
	\selfconnectionleft{D};
	\selfconnectionleft{E};
	\selfconnectionright{F};

\end{tikzpicture}
\caption{A data structure $\AA$ and  $\gaifmanish{\AA}$.}
\label{fig:abstract-a}
	\end{subfigure}
	\unskip\ \vrule\ \hspace{1em}
	\begin{subfigure}[b]{0.45\textwidth}
\begin{tikzpicture}[node distance=2cm]
	\node [dataone, label=below left:$a$]                (A)    {8};
	\node [dataone, above left of=A,xshift=-1em,label=below:$b$]    (B)    {3};
	\node [dataone, above right of=A,xshift=1em,label=below:$c$]    (C)    {3};
	\node [dataone, below left of=A,label=below:$d$]    (D)    {9};
	\node [dataone, above right of=B,xshift=1em, label=below:$e$]
    (E)    {10};
    \node [dataone, below right of=A, label=below:$f$]    (F)    {7};
    \node [right of=C,yshift=-3em]    (G)    {$\begin{array}{l}\uP{a[1,1]}=\{a,b\} \\
                                     \uP{a[2,2]}=\{a,c\}\\
                                     \uP{a[1,2]}=\emptyset \\
                                   \uP{a[2,1]}=\{f\}\\\end{array}$};

\end{tikzpicture}
\caption{$\sem{\AA}_{(a)}$.}
\label{fig:abstract-b}
	\end{subfigure}
	\caption{
\label{fig:abstract}}
\end{figure*}

We now explain formally how we build $\AAas$. Let $\Udeci{n}=\{\udd{p}{i}{j}\mid p\in\{1,\ldots,n\}, i,j\in\{1,2\}\}$ be a set of new unary predicates and $\Unaryp = \Unary \cup \Udeci{n}$.
For every element $b\in A$, the predicates in $\Udeci{n}$ are used to keep track of the relation between the data values of $b$ and the one of $a_1,\ldots,a_n$ in $\AA$.
Formally, we define $\uP{\udd{p}{i}{j}}=\{b\in A\mid \AA\models \rels{i}{j}{a_p}{b}\}$.
We now define a data function $f:A\to \N$.
We recall for this matter that $\Valuessub{\AA}{\tuple{a}} = \{f_1(a_1),f_2(a_1),\ldots,f_1(a_n),f_2(a_n)\}$ and let $\inj:A\to\N\setminus \Values{\AA}$ be an injection. For every $b \in A$, we set:
\[
	f(b) = \begin{cases}
		f_2(b) \text{ if } f_1(b)\in \Valuessub{\AA}{\tuple{a}} \text{ and } f_2(b)\notin \Valuessub{\AA}{\tuple{a}}\\
		f_1(b) \text{ if } f_1(b)\notin \Valuessub{\AA}{\tuple{a}} \text{ and } f_2(b)\in \Valuessub{\AA}{\tuple{a}}\\
		\inj(b) \text{ otherwise}
    	   \end{cases}
\]
Hence depending if $f_1(b)$ or $f_2(b)$ is in $\Valuessub{\AA}{\tuple{a}}$, it splits the elements of $\AA$ in four categories.
If $f_1(b)$ and $f_2(b)$ are in $\Valuessub{\AA}{\tuple{a}}$, the predicates in $\Udeci{n}$ allow us to retrieve all the data values of $b$. 
Given $j\in\{1,2\}$, if $f_j(b)$ is in $\Valuessub{\AA}{\tuple{a}}$ but $f_{3-j}(b)$ is not, the new predicates will give us the $j$-th data value of $b$ and we have to keep track of the $(3-j)$-th one, so we save it in $f(b)$.
Lastly, if neither $f_1(b)$ nor $f_2(b)$ is in $\Valuessub{\AA}{\tuple{a}}$, we will never be able to see the data values of $b$ in $\phi_{q_f}$ (thanks to Lemma \ref{lem:shape-balls}), so they do not matter to us. Finally, we have  $\AAas = (A, (\uP{\unary})_{\unary\in\Unaryp}, f) $. Figure \ref{fig:abstract-b} provides an example of  $\Valuessub{\AA}{\tuple{a}}$ for the data structures depicted on Figure \ref{fig:abstract-a} and $\tuple{a}=(a)$.

The next lemma formalizes the connection existing between $\AA$ and
$\AAas$ with $\tuple{a} = (a_1,\ldots,a_n)$.

\begin{lemma}\label{lem:r2dv2-semantique}
	Let $b,c\in A$ and $j,k\in\{1,2\}$ and $p\in\{1,\ldots,n\}$. The following statements then hold.
	\begin{enumerate}
    \item If $(b,j)\in\Ball{1}{a_p}{\AA}$ and $(c,k)\in\Ball{1}{a_p}{\AA}$ then $\vprojr{\AA}{a_p}{2}\models\rels{j}{k}{b}{c}$ iff there is $i\in\{1,2\}$ s.t. $b \in \uP{\udd{p}{i}{j}}$ and $c \in \uP{\udd{p}{i}{k}}$.
		\item If $(b,j)\in\Ball{2}{a_p}{\AA}\setminus\Ball{1}{a_p}{\AA}$ and $(c,k)\in\Ball{1}{a_p}{\AA}$ then $\vprojr{\AA}{a_p}{2}\nvDash\rels{j}{k}{b}{c}$ 
		\item If $(b,j),(c,k) \in\Ball{2}{a_p}{\AA}\setminus\Ball{1}{a_p}{\AA}$ then $\vprojr{\AA}{a_p}{2}\models\rels{j}{k}{b}{c}$ iff either                $\relsaa{1}{1}{\AAas}{b}{c}$ or there exists $p' \in \{1,\ldots,n\}$ and $\ell \in \{1,2\}$ such that $b \in \uP{\udd{p'}{\ell}{j}}$ and $c \in \uP{\udd{p'}{\ell}{k}}$ .
		\item If $(b,j)\notin\Ball{2}{a_p}{\AA}$ and $(c,k)\in\Ball{2}{a_p}{\AA}$ then $\vprojr{\AA}{a_p}{2}\nvDash\rels{j}{k}{b}{c}$
		\item If $(b,j)\notin\Ball{2}{a_p}{\AA}$ and $(c,k)\notin\Ball{2}{a_p}{\AA}$ then $\vprojr{\AA}{a_p}{2}\models\rels{j}{k}{b}{c}$ iff $b=c$ and $j=k$.
	\end{enumerate}
\end{lemma}

\begin{proof}
	We suppose that $\vprojr{\AA}{a_p}{2} = (A,(\uP{\unary})_\unary,\fp_1,\fp_2)$.
	\begin{enumerate}
		\item Assume that $(b,j)\in\Ball{1}{a_p}{\AA}$ and $(c,k)\in\Ball{1}{a_p}{\AA}$.
			It implies that $\fp_j(b)=f_j(b)$ and $\fp_k(c)=f_k(c)$.
			Then assume that $\vprojr{\AA}{a_p}{2}\models\rels{j}{k}{b}{c}$. 
			As $(b,j)\in\Ball{1}{a_p}{\AA}$, thanks to Lemma \ref{lem:shape-balls}.1 it means that there is a $i\in\{1,2\}$ such that $\relsaa{i}{j}{\AA}{a_p}{b}$.
			So we have $f_k(c)=\fp_k(c)=\fp_j(b)=f_j(b)=f_i(a_p)$, that is $\relsaa{i}{k}{\AA}{a_p}{c}$. Hence by definition, $b \in \uP{\udd{p}{i}{j}}$ and $c \in \uP{\udd{p}{i}{k}}$.
			Conversely, let $i\in\{1,2\}$ such that $b \in \uP{\udd{p}{i}{j}}$ and $c \in \uP{\udd{p}{i}{k}}$. This means that $\relsaa{i}{j}{\AA}{a_p}{b}$ and $\relsaa{i}{k}{\AA}{a_p}{c}$.
			So $\fp_j(b)=f_j(b)=f_i(a_p)=f_k(c)=\fp_k(c)$, that is $\vprojr{\AA}{a_p}{2}\models\rels{j}{k}{b}{c}$. 
		\item Assume that $(b,j)\in\Ball{2}{a_p}{\AA}\setminus\Ball{1}{a_p}{\AA}$ and $(c,k)\in\Ball{1}{a_p}{\AA}$.
			It implies that $\fp_j(b)=f_j(b)$ and $\fp_k(c)=f_k(c)$.
			Thanks to Lemma \ref{lem:shape-balls}.1, $(c,k)\in\Ball{1}{a_p}{\AA}$ implies that $f_k(c)\in\{f_1(a_p),f_2(a_p)\}$ and $(b,j)\notin\Ball{1}{a_p}{\AA}$ implies that $f_j(b)\notin\{f_1(a_p),f_2(a_p)\}$.
			So  $\vprojr{\AA}{a_p}{2}\not \models\rels{j}{k}{b}{c}$.
		\item Assume that $(b,j), (c,k) \in\Ball{2}{a_p}{\AA}\setminus\Ball{1}{a_p}{\AA}$. As previously, we have that $f_j(b)\notin\{f_1(a_p),f_2(a_p)\}$ and $f_k(c)\notin\{f_1(a_p),f_2(a_p)\}$, and thanks to Lemma \ref{lem:shape-balls}.2, we have $f_{3-j}(b) \in \{f_1(a_p),f_2(a_p)\}$ and $f_{3-k}(b) \in \{f_1(a_p),f_2(a_p)\}$. There is then two cases:
    \begin{itemize}
    \item Suppose there does not exists $p' \in \{1,\ldots,n\}$ such that $f_{j}(b) \in \{f_1(a_{p'}),f_2(a_{p'})\}$ .This allows us to deduce that $\fp_j(b)=f_j(b)=f(b)$ and $\fp_k(c)=f_k(c)$. If $\vprojr{\AA}{a_p}{2}\models\rels{j}{k}{b}{c}$, then necessarily there does not exists $p' \in \{1,\ldots,n\}$ such that $f_{k}(c) \in \{f_1(a_{p'}),f_2(a_{p'})\}$ so we have $\fp_k(c)=f_k(c)=f(c)$ and  $f(b)=f(c)$, consequently $\relsaa{1}{1}{\AAas}{b}{c}$. Similarly assume that  $\relsaa{1}{1}{\AAas}{b}{c}$, this means that $f(b)=f(c)$ and either $b=c$ and $k=j$ or $b \neq c$ and by injectivity of $f$,we have $f_j(b)=f(b)=f_k(c)$. This allows us to deduce that $\vprojr{\AA}{a_p}{2}\models\rels{j}{k}{b}{c}$.
    \item  If there exists $p' \in \{1,\ldots,n\}$ such that $f_{j}(b) = f_\ell(a_{p'})$ for some $\ell \in \{1,2\}$. Then we have $b \in       \uP{\udd{p'}{\ell}{j}}$. Consequently, we have $\vprojr{\AA}{a_p}{2}\models\rels{j}{k}{b}{c}$ iff $c \in \uP{\udd{p'}{\ell}{k}}$.
   \end{itemize}
		\item We prove the case 4 and 5 at the same time.
			Assume that $(b,j)\notin\Ball{2}{a_p}{\AA}$.
			It means that in order to have $\fp_j(b)=\fp_k(c)$, we must have $(b,j)=(c,k)$.
			So if $(c,k)\in\Ball{2}{a_p}{\AA}$, we can not have $\vprojr{\AA}{a_p}{2}\models\rels{j}{k}{b}{c}$ which ends case 4.
			And if $(c,k)\notin\Ball{2}{a_p}{\AA}$, we have that $\vprojr{\AA}{a_p}{2}\models\rels{j}{k}{b}{c}$ iff $b=c$ and $j=k$.
	\end{enumerate}
\end{proof}

We shall now see how we translate the formula
$\phi_{qf}(x_1,\ldots,x_n)$ into a formula
$\phit{\phi_{qf}}(x_1,\ldots,x_n)$ in $\ndFO{1}{\Unaryp}$ such that $\AA$ satisfies $\phi_{qf}(a_1,\ldots,a_n)$ if, and only if, $\AAas$ satisfies $\phit{\phi_{qf}}(a_1,\ldots,a_n)$. Thanks to the previous lemma we know that if $\vprojr{\AA}{a_p}{2}\models\rels{j}{k}{b}{c}$ then $(b,j)$ and $(c,k)$ must belong to the  same set among $\Ball{1}{a_p}{\AA}$, $\Ball{2}{a_p}{\AA}\setminus\Ball{1}{a_p}{\AA}$ and $\comp{\Ball{2}{a_p}{\AA}}$ and we can test in $\AAas$ whether $(b,j)$ is a member of  $\Ball{1}{a_p}{\AA}$ or $\Ball{2}{a_p}{\AA}$.
Indeed, thanks to Lemmas \ref{lem:shape-balls}.1 and \ref{lem:shape-balls}.2, we have $(b,j) \in \Ball{1}{a_p}{\AA}$ iff $b\in\bigcup_{i=1,2}\uP{\udd{p}{i}{j}}$ and $(b,j) \in \Ball{2}{a_p}{\AA}$ iff $b\in\bigcup_{i=1,2}^{j'=1,2} \uP{\udd{p}{i}{j'}}$. This reasoning leads to the  following formulas in $\ndFO{1}{\Unaryp}$ with $p \in \{1,\ldots,n\}$ and $j \in \{1,2\}$:
\begin{itemize}
\item $\phiBun{j}(y) := \udd{p}{1}{j}(y) \ou \udd{p}{2}{j}(y)$ to test if the $j$-th field of an element belongs to $\Ball{1}{a_p}{\AA}$
\item $\phiBdeux(y) := \phiBun{1}(y) \ou \phiBun{2}(y)$ to test if a field of an element belongs to $\Ball{2}{a_p}{\AA}$
\item $\phiBdsu{j}(y) := \phiBdeux(y) \et \neg\phiBun{j}(y)$ to test that the $j$-th field of an element belongs to $\Ball{2}{a_p}{\AA}\setminus\Ball{1}{a_p}{\AA}$
\end{itemize}
  
We shall now present how we use these formulas to translate atomic formulas of the form  $\rels{j}{k}{y}{z}$ under some $\locformr{-}{x_p}{2}$. For this matter, we rely on the three following formulas of $\ndFO{1}{\Unaryp}$:
\begin{itemize}
\item The first formula asks  for $(y,j)$ and $(z,k)$ to be in $\Ball{1}{a_p}{1}$ (where here we abuse notations, using variables for the elements they represent) and for these two data values to coincide with one data value of $a_p$, it corresponds to Lemma \ref{lem:r2dv2-semantique}.1:
  $$
  \phiun(y,z) := \phiBun{j}(y) \et \phiBun{k}(z) \et \Ou_ {i=1,2}\udd{p}{i}{j}(y)\et\udd{p}{i}{k}(z) 
  $$
\item The second formula asks for $(y,j)$ and $(z,k)$ to be in $\Ball{2}{a_p}{\AA}\setminus\Ball{1}{a_p}{\AA}$ and checks either whether the data values of $y$ and $z$ in $\AAas$ are equal or whether there exist $p'$ and $\ell$ such that $y$ belongs to $\udd{p'}{\ell}{j}(y)$ and $z$ belongs to $\udd{p'}{\ell}{k}(z)$, it corresponds to Lemma \ref{lem:r2dv2-semantique}.3:
  $$
  \phideux(y,z) := \phiBdsu{j}(y) \et \phiBdsu{k}(z) \et \big (y\sim z \ou\big(\Ou^n_{p'=1}\Ou^2_ {\ell=1}\udd{p'}{\ell}{j}(y)\et\udd{p'}{\ell}{k}(z)\big)\big)  
  $$
\item The third formula asks for $(y,j)$ and $(z,k)$ to not belong to $\Ball{2}{a_p}{\AA}$ and for $y=z$, it corresponds to Lemma \ref{lem:r2dv2-semantique}.5:
  $$
  \phitrois(y,z) := \begin{cases}
		                                	\neg \phiBdeux(y) \et \neg\phiBdeux(z) \et y=z  &\text{ if } j=k \\
																			\bot  &\text{ otherwise}
		                                \end{cases}
  $$
\end{itemize}

	Finally, here is the inductive definition of the translation $\T{-}$ which uses sub transformations $\Tp{-}$ in order to remember the centre of the ball and leads to the construction of $\phit{\phi_{qf}}(x_1,\ldots,x_n)$:
\[	\begin{array}{rcl}
		\T{\phi\ou\phi'} &=& \T{\phi} \ou \T{\phi'}\\
		\T{x_p=x_p'}  &=& x_p=x_p'            \\      
		\T{\neg\phi} &=& \neg\T{\phi}       \\   
      \T{\locformr{\psi}{x_p}{2}}  &=& \Tp{\psi}  \\
      \Tp{\rels{j}{k}{y}{z}} &=&\phiun(y,z) \ou \phideux(y,z) \ou \phitrois(y,z)\\
      \Tp{\unary(x)} &=& \unary(x)  \\
      \Tp{x=y} &=& x=y \\
      \Tp{\phi\ou\phi'}&=& \Tp{\phi} \ou \Tp{\phi'} \\
      \Tp{\neg\phi} &=& \neg\Tp{\phi}\\
      \Tp{\exists x. \phi} &=& \exists x.\Tp{\phi}\\
	\end{array}\]

\begin{lemma} \label{lem:correct}
	We have $\AA\models\phi_{qf}(\tuple{a})$ iff $\AAas\models\phit{\phi_{qf}}(\tuple{a})$.
\end{lemma}
\begin{proof}
  Because of the inductive definition of $\T{\phi}$ and that only the atomic formulas $\rels{j}{k}{y}{z}$ change, we only have to prove that given $b,c\in A$, we have $\vprojr{\AA}{a_p}{2}\models\rels{j}{k}{b}{c}$ iff $\AAas\models \Tp{\rels{j}{k}{y}{z}}(b,c)$.
  
	We first suppose that $\vprojr{\AA}{a_p}{2}\models\rels{j}{k}{b}{c}$.
    Using Lemma \ref{lem:r2dv2-semantique},  it implies that $(b,j)$ and $(c,k)$ belong to same set between $\Ball{1}{a_p}{\AA}$, $\Ball{2}{a_p}{\AA} \setminus \Ball{1}{a_p}{\AA}$ and $\comp{\Ball{2}{a_p}{\AA}}$. We proceed by a case analysis.
    \begin{itemize}
	\item If $(b,j),(c,k)\in\Ball{1}{a_p}{\AA}$ then by lemma \ref{lem:r2dv2-semantique}.1 we have that $\AAas\models\phiun(b,c)$ and thus $\AAas\models \Tp{\rels{j}{k}{y}{z}}(b,c)$.
      
	\item If $(b,j),(c,k)\in\Ball{2}{a_p}{\AA} \setminus \Ball{1}{a_p}{\AA}$ then by lemma \ref{lem:r2dv2-semantique}.3 we have that $\AAas\models\phideux(b,c)$ and thus $\AAas\models \Tp{\rels{j}{k}{y}{z}}(b,c)$.
	\item If $(b,j),(c,k)\in\comp{\Ball{2}{a_p}{\AA}}$ then by lemma \ref{lem:r2dv2-semantique}.5 we have that $\AAas\models\phitrois(b,c)$ and thus $\AAas\models \Tp{\rels{j}{k}{y}{z}}(b,c)$.
    \end{itemize}

 We now suppose that $\AAas\models \Tp{\rels{j}{k}{y}{z}}(b,c)$.
	It means that $\AAas$ satisfies at least $\phiun(b,c)$, $\phideux(b,c)$ or $\phitrois(b,c)$.
	If $\AAas\models\phiun(b,c)$, it implies that $(b,j)$ and $(c,k)$ are in $\Ball{1}{a_p}{\AA}$, and we can then apply lemma \ref{lem:r2dv2-semantique}.1 to deduce that $\vprojr{\AA}{a_p}{2}\models\rels{j}{k}{b}{c}$.
	If $\AAas\models\phideux(b,c)$, it implies that $(b,j)$ and $(c,k)$ are in $\Ball{2}{a_p}{\AA} \setminus \Ball{1}{a_p}{\AA}$, and we can then apply lemma \ref{lem:r2dv2-semantique}.3 to deduce that $\vprojr{\AA}{a_p}{2}\models\rels{j}{k}{b}{c}$.
	If $\AAas\models\phitrois(b,c)$, it implies that $(b,j)$ and $(c,k)$ are in $\comp{\Ball{2}{a_p}{\AA}}$, and we can then apply lemma \ref{lem:r2dv2-semantique}.5 to deduce that $\vprojr{\AA}{a_p}{2}\models\rels{j}{k}{b}{c}$.
  \end{proof}
  
  \medskip

To provide a reduction from  $\nDataSat{\eFOr{2}}{2}$ to
$\nDataSat{\ndFOr}{1}$, having the formula $\phit{\phi_{qf}}(x_1,\ldots,x_n)$
is not enough because to use the result of the previous
Lemma, we need to ensure that there exists a model $\BB$ and a tuple
of elements $(a_1,\ldots,a_n)$ such that $\BB \models\
\phit{\phi_{qf}}(a_1,\ldots,a_n)$ and as well that there exists
$\AA\in \nData{2}{\Unary}$ such that $ \BB = \AAas$. We explain now how we
can ensure this last point.

 Now, we want to characterize the structures of the form $\AAas$.
Given $\BB =
(A,(\uP{\unary})_{\unary\in\Unaryp},f)\in\nData{1}{\Unaryp}$ and
$\tuple{a}\in A$, we say that $(\BB,\tuple{a})$ is \emph{well formed}
iff there exists a structure $\AA\in \nData{2}{\Unary}$ such that $ \BB
= \AAas$. Hence $(\BB,\tuple{a})$ is \emph{well formed} iff there
exist two  functions $f_1,f_2:A\to\N$ such that $\AAas=\sem{(A,(\uP{\unary})_{\unary\in\Unary}, f_1,f_2)}_{\tuple{a}}$.
We state three properties on $(\BB,\tuple{a})$, and we will show that they characterize being well formed.
\begin{enumerate}
	\item (Transitivity) For all $b,c\in A$, $p,q \in\{1,\ldots,n\}$,
      $i,j,k,\ell \in\{1,2\}$ if $b\in\uP{\udd{p}{i}{j}}$, $c\in\uP{\udd{p}{i}{\ell}}$ and  $b\in\uP{\udd{q}{k}{j}}$ then $c\in\uP{\udd{q}{k}{\ell}}$.
	\item (Reflexivity) For all $p$ and $i$, we have $a_p\in\uP{\udd{p}{i}{i}}$
	\item (Uniqueness) For all $b\in A$, if $b\in\bigcap_{j=1,2}\bigcup_{p=1,\ldots,n}^{i=1,2} \uP{\udd{p}{i}{j}}$ or $b\notin\bigcup_{j=1,2}\bigcup_{p=1,\ldots,n}^{i=1,2} \uP{\udd{p}{i}{j}}$ then for any $c\in B$ such that $f(c)=f(b)$ we have $c=b$.
\end{enumerate}
Each property can be expressed by a first order logic formula, which
we respectively name $\phitran$, $\phirefl$ and $\phiuniq$ and  we
denote by $\phiwf$ their conjunction:
$$
  \begin{array}{ll}
\phitran &= \forall y \forall z.\Et_{p,q=1}^{n}\Et_{i,j,k,\ell=1}^2 \Big(\udd{p}{i}{j}(y) \et \udd{p}{i}{\ell}(z) \et \udd{q}{k}{j}(y) \donc \udd{q}{k}{\ell}(z)\Big) \\
\phirefl(x_1,\ldots,x_n)  &=\Et_{p=1}^n\Et_{i=1}^2	 \udd{p}{i}{i}(x_p) \\
\phiuniq &= \forall y. \Big(\Et_{j=1}^2 \Ou^n_{p=1} \Ou_{i=1}^2 \udd{p}{i}{j}(y) \ou \Et_{j=1}^2 \Et^n_{p=1}\Et^2_{i=1} \neg\udd{p}{i}{j}(y)\Big) \donc (\forall z. y\sim z \donc y=z)\\
\phiwf(x_1,\ldots,x_n) &=\phitran \et \phirefl(x_1,\ldots,x_n) \et
                         \phiuniq
  \end{array}
  $$

The next lemma expresses that the formula $\phiwf$ allows to
characterise precisely the $1$-data structures in $\nData{1}{\Unaryp}$
which are well-formed.

\begin{lemma}\label{lem:well-formed}
	Let $\BB\in\nData{1}{\Unaryp}$ and $a_1,\ldots,a_n$ elements of
    $\BB$, then $(\BB,\tuple{a})$ is well formed iff $\BB\models\phiwf(\tuple{a})$.
  \end{lemma}

  \begin{proof}
    First, if $(\BB,\tuple{a})$ is well formed, then there there exists
$\AA\in \nData{2}{\Unary}$ such that $ \BB = \AAas$ and by
construction we have $\AAas \models\phiwf(\tuple{a})$. We now suppose
that $\BB=(A,(\uP{\unary})_{\unary\in\Unaryp},f)$ and $\BB\models\phiwf(\tuple{a})$.	
	In order to define the functions $f_1,f_2:A\to\N$, we  need
    to introduce some objects.

    We first define a function $g :
    \{1,\ldots,n\} \times \{1,2\} \to \N\setminus \im{f}$ (where
    $\im{f}$ is the image of $f$ in $\BB$) which
      verifies the following properties:
      \begin{itemize}
        \item for all $p \in \{1,\ldots,n\}$ and $i \in \{1,2\}$, we
          have 
          $a_p \in \uP{\udd{p}{i}{3-i}} $ iff $g(p,1)=g(p,2)$;
        \item for all $p, q \in \{1,\ldots,n\}$ and $i,j \in \{1,2\}$,
          we have  $a_q \in \uP{\udd{p}{i}{j}} $ iff $g(p,i)=g(q,j)$.
        \end{itemize}
We use this function to fix the two data values carried by the
elements in $\{a_1,\ldots,a_m\}$. We now explain why this function is
well founded, it is due to the fact
that $\BB\models\phitran \et \phirefl(a_1,\ldots,a_n)$. In fact, since
$\BB \models \phirefl(a_1,\ldots,a_n)$, we have for all $p \in
\{1,\ldots,n\}$ and $i \in \{1,2\}$, $a_p \in \uP{\udd{p}{i}{i}}
$. Furthermore if $a_p \in \uP{\udd{p}{i}{j}}$ then $a_p \in
\uP{\udd{p}{j}{i}}$ thanks to the formula $\phitran$; indeed since we
have $a_p \in  \uP{\udd{p}{i}{j}}$ and $a_p \in  \uP{\udd{p}{i}{i}}$
and  $a_p \in  \uP{\udd{p}{j}{j}}$, we obtain $a_p \in
\uP{\udd{p}{j}{i}}$. Next, we also have that if $a_q \in
\uP{\udd{p}{i}{j}}$ then $a_p \in
\uP{\udd{q}{j}{i}}$ again thanks to $\phitran$;  indeed since we
have $a_q \in  \uP{\udd{p}{i}{j}}$ and $a_p \in  \uP{\udd{p}{i}{i}}$
and  $a_q \in  \uP{\udd{q}{j}{j}}$, we obtain $a_p \in
\uP{\udd{q}{j}{i}}$.

We also need a natural $\dout$ belonging to $\N\setminus
    (\im{g}\cup\im{f})$. For $j \in
    \{1,2\}$, we define $f_j$ as follows for all $b \in A$:
	\[f_j(b) = \left\{\begin{array}{ll}
		g(p,i) & \text{if for some } p,i \text{ we have } b\in\uP{\udd{p}{i}{j}} \\
		f(b) &\text{if for all $p,i$ we have $b\notin\uP{\udd{p}{i}{j}}$ and for some $p,i$ we have $b\in\uP{\udd{p}{i}{3-j}}$} \\
						 \dout &\text{if for all $p,i,j'$, we have } b\notin\uP{\udd{p}{i}{j'}}
	           \end{array}\right.
	\]

Here again, we can show that since $\BB\models\phitran \et
\phirefl(a_1,\ldots,a_n)$, the functions $f_1$ and $f_2$ are well
founded. Indeed, assume that  $b\in\uP{\udd{p}{i}{j}} \cap
  \uP{\udd{q}{k}{j}}$, then we have necessarily that
  $g(p,i)=g(q,k)$. For this we need to show that $a_p \in
  \udd{q}{k}{i}$ and we use again the formula $\phitran$. This can be
  obtained because we have $b\in\uP{\udd{p}{i}{j}}$ and
  $a_p\in\uP{\udd{p}{i}{i}}$ and $b \in \uP{\udd{q}{k}{j}}$.

 We   then define $\AA$ as the $2$-data-structures
$(A,(P_{\unary})_{\unary \in \Unary},\f{1},\f{2})$. It remains to
prove that $\BB = \AAas$. 

First, note  that for all  $b\in A$, $p \in \{1,\ldots,n\}$ and
$i,j\in\{1,2\}$, we have $b\in\uP{\udd{p}{i}{j}}$ iff
$\relsaa{i}{j}{\AA}{a_p}{b}$. Indeed, we have $b\in\uP{\udd{p}{i}{j}}$,
we have that $f_j(b)=g(p,i)$ and since $a_p \in \uP{\udd{p}{i}{j}}$ we
have as well that $f_i(a_p)=g(p,i)$, as a consequence
$\relsaa{i}{j}{\AA}{a_p}{b}$. In the other direction, if
$\relsaa{i}{j}{\AA}{a_p}{b}$, it means that $f_j(b)=f_i(a_p)=g(p,i)$
and thus  $b\in\uP{\udd{p}{i}{j}}$. Now to have $\BB = \AAas$, one has
only to be careful in the choice of  function $\inj$ 
while building $\AAas$. We recall that this function is injective and is
used to give a value to the elements $b \in A$ such that neither
$f_1(b)\in \Valuessub{\AA}{\tuple{a}} \text{ and } f_2(b)\notin
\Valuessub{\AA}{\tuple{a}}$ nor $ f_1(b)\notin
\Valuessub{\AA}{\tuple{a}} \text{ and } f_2(b)\in
\Valuessub{\AA}{\tuple{a}}$. For these elements, we make $\inj$
matches with the function $f$ and the fact that we define an injection
is guaranteed by the formula $\phiuniq$.
\end{proof}

Using the results of Lemma \ref{lem:correct} and
\ref{lem:well-formed}, we deduce that the formula $\phi=\exists x_1\ldots\exists
x_n.\phi_{qf}(x_1,\ldots,x_n)$ of $\eFO{2}{\Unary}{2}$ is satisfiable
iff the formula $\psi=\exists x_1\ldots\exists
x_n.\phit{\phi_{qf}}(x_1,\ldots,x_n) \wedge \phiwf(x_1,\ldots,x_n) $
is satisfiable. Note that $\psi$ can be built in polynomial time from
$\phi$ and that it belongs to $\ndFO{1}{\Unaryp}$. Hence, thanks to
Theorem \ref{thm:1fo}, we obtain that $\nDataSat{\eFOr{2}}{2}$ is in
\textsc{N2EXP}. 

We can as well obtain a matching lower bound thanks to a
reduction from  $\nDataSat{\ndFOr}{1}$. For this matter we rely on two
crucial points. First in the formulas of $\eFO{2}{\Unary}{2}$, there is no
restriction on the use of quantifiers for the formulas located under the scope
of the $\locformr{\cdot}{x}{2}$ modality and consequently we can write
inside this modality a formula of $\ndFO{1}{\Unary}$ without any
modification. Second we can extend a
model $\ndFO{1}{\Unary}$ into a $2$-data structure such that all
elements and their values are located in the same radius-$2$-ball by adding everywhere a second
data value equal to $0$. More formally, let $\phi$ be
a formula in $\ndFO{1}{\Unary}$ and consider the formula $\exists
x.\locformr{\phi}{x}{2}$ where we interpret $\phi$ over $2$-data
structures (this formula simply never mentions the values located in the second
fields). We have then the following lemma.

\begin{lemma} \label{lem:hardness-radius2-2}
 There exists $\AA \in \nData{1}{\Unary}$ such that $\AA \models \phi$
 if and only if there exists $\BB \in  \nData{2}{\Unary}$ such that
 $\BB \models \exists
x.\locformr{\phi}{x}{2}$.
\end{lemma}

\begin{proof}
 Assume that there exists  $\AA=(A,(P_{\unary})_{\unary \in
   \Unary},\f{1})$ in $\nData{1}{\Unary}$ such that  $\AA \models
 \phi$. Consider the $2$-data structure $\BB=(A,(P_{\unary})_{\unary \in
   \Unary},\f{1},\f{2})$ such that $\f{2}(a)=0$ for all $a\in
 A$. Let $a \in A$. It is clear that we have $\vprojr{\BB}{a}{2}=\BB$
 and that $\vprojr{\BB}{a}{2} \models \phi$ (because $\AA \models
 \phi$ and $\phi$ never mentions the second values of the elements
 since it is a formula in $\ndFO{1}{\Unary}$ ). Consequently $\BB \models \exists
 x.\locformr{\phi}{x}{2}$.

 Assume now that there exists $\BB=(A,(P_{\unary})_{\unary \in
   \Unary},\f{1},\f{2})$ in $ \nData{2}{\Unary}$ such that  $\BB \models \exists
x.\locformr{\phi}{x}{2}$. Hence there exists $a \in A$ such that
$\vprojr{\BB}{a}{2} \models \phi$, but then by forgetting the second
value in $\vprojr{\BB}{a}{2}$ we obtain a model in $\nData{1}{\Unary}$
which satisfies $\phi$.
\end{proof}
  
Since $\nDataSat{\ndFOr}{1}$ is
\textsc{N2EXP}-hard (see Theorem \ref{thm:1fo}), we obtain the desired lower bound. 

\begin{theorem}\label{thm:radius2-2}
	The problem $\nDataSat{\eFOr{2}}{2}$ is \textsc{N2EXP}-complete.
\end{theorem}

\subsection{Balls of radius 1 and any number of data values }

Let $D \geq 1$. We first show that  $\nDataSat{\eFOr{1}}{D}$ is in
\textsc{NEXP} by providing a reduction towards
$\nDataSat{\ndFOr}{0}$. This reduction uses the characterisation of
the radius-1-ball provided by Lemma  \ref{lem:shape-balls} and is very
similar to the reduction provided in the previous section.  In fact,
 for an element $b$ located in the radius-1-ball of another
element $a$, we use extra unary predicates to explicit which are the
values of $b$ that are common with the  values of $a$. We provide here
the main step of this reduction whose proof follows the  same line as
the one of Theorem \ref{thm:radius2-2}.

We consider a formula $\phi=\exists x_1\ldots\exists
x_n.\phi_{qf}(x_1,\ldots,x_n)$ of $\eFO{D}{\Unary}{1}$ in prenex normal
form, i.e., such that $\phi_{qf}(x_1,\ldots,x_n)\in\qfFO{D}{\Unary}{1}$. We
know that there is a structure $\AA=(A,(P_{\unary})_{\unary \in
  \Unary},\linebreak[0]\f{1},\f{2},\ldots,\f{D})$ in
$\nData{D}{\Unary}$ such that $\AA\models\phi$ if and only if there
are $a_1,\ldots,a_n \in A $ such that
$\AA\models\phi_{qf}(a_1,\ldots,a_n)$. Let then $\AA=(A,(P_{\unary})_{\unary \in \Unary},\f{1},\f{2},\ldots,\f{D})$ in $\nData{D}{\Unary}$ and a tuple $\tuple{a} = (a_1,\ldots,a_n)$ of elements in $A^n$. Let $\Omega_n=\{\udd{p}{i}{j}\mid p\in\{1,\ldots,n\}, i,j\in\{1,\ldots,D\}\}$ be a set of new unary predicates and $\Unaryp = \Unary \cup \Omega_n$.
For every element $b\in A$, the predicates in $\Omega_n$ are used to keep track of the relation between the data values of $b$ and the one of $a_1,\ldots,a_n$ in $\AA$.
Formally, we have $\uP{\udd{p}{i}{j}}=\{b\in A\mid \AA\models \rels{i}{j}{a_p}{b}\}$.
Finally, we build the $0$-data-structure
$\sem{\AA}'_{\tuple{a}}= (A, (\uP{\unary})_{\unary\in\Unaryp})
$. Similarly to Lemma \ref{lem:r2dv2-semantique}, we have the
following connection between $\AA$ and $\sem{\AA}'_{\tuple{a}}$.

\begin{lemma}\label{lem:r1-semantique}
	Let $b,c\in A$ and $j,k\in\{1,\ldots,D\}$ and $p\in\{1,\ldots,n\}$. The following statements hold:
	\begin{enumerate}
    \item If $(b,j)\in\Ball{1}{a_p}{\AA}$ and $(c,k)\in\Ball{1}{a_p}{\AA}$ then $\vprojr{\AA}{a_p}{1}\models\rels{j}{k}{b}{c}$ iff there is $i\in\{1,2\}$ s.t. $b \in \uP{\udd{p}{i}{j}}$ and $c \in \uP{\udd{p}{i}{k}}$.
		\item If $(b,j)\notin\Ball{1}{a_p}{\AA}$ and $(c,k)\in\Ball{1}{a_p}{\AA}$ then $\vprojr{\AA}{a_p}{1}\nvDash\rels{j}{k}{b}{c}$
		\item If $(b,j)\notin\Ball{1}{a_p}{\AA}$ and $(c,k)\notin\Ball{1}{a_p}{\AA}$ then $\vprojr{\AA}{a_p}{1}\models\rels{j}{k}{b}{c}$ iff $b=c$ and $j=k$.
	\end{enumerate}
\end{lemma}

We shall now see how we translate the formula
$\phi_{qf}(x_1,\ldots,x_n)$ into a formula
$\phit{\phi_{qf}}'(x_1,\ldots,x_n)$ in $\ndFO{0}{\Unaryp}$ such that $\AA$
satisfies $\phi_{qf}(a_1,\ldots,a_n)$ if, and only if,
$\sem{\AA}'_{\tuple{a}}$ satisfies
$\phit{\phi_{qf}}(a_1,\ldots,a_n)$. As in the previous section, we
introduce the  following formula in $\ndFO{0}{\Unaryp}$ with $p \in
\{1,\ldots,n\}$ and $j \in \{1,\ldots,D\}$ to test if the $j$-th field of an element belongs to $\Ball{1}{a_p}{\AA}$:
$$
\phiBun{j}(y) := \bigvee_{i \in \{1,\ldots,D\}}\udd{p}{i}{j}(y)
$$

We now present how we  translate atomic formulas of the form  $\rels{j}{k}{y}{z}$ under some $\locformr{-}{x_p}{1}$. For this matter, we rely on two formulas of $\ndFO{0}{\Unaryp}$ which can be described as follows:
\begin{itemize}
\item The first formula asks  for $(y,j)$ and $(z,k)$ to be in $\Ball{1}{a_p}{1}$ (here we abuse notations, using variables for the elements they represent) and for these two data values to coincide with one data value of $a_p$, it corresponds to Lemma \ref{lem:r1-semantique}.1:
  $$
  \psiun(y,z) := \phiBun{j}(y) \et \phiBun{k}(z) \et \Ou^D_ {i=1}\udd{p}{i}{j}(y)\et\udd{p}{i}{k}(z) 
  $$
\item The second formula asks for $(y,j)$ and $(z,k)$ to not belong to $\Ball{1}{a_p}{\AA}$ and for $y=z$, it corresponds to Lemma \ref{lem:r1-semantique}.3:
  $$
  \psideux(y,z) := \begin{cases}
		                                	\bigwedge^D_{i=1} (\neg
                                            \phiBun{i}(y) \wedge \neg  \phiBun{i}(z))  \et y=z  &\text{ if } j=k \\
																			\bot  &\text{ otherwise}
		                                \end{cases}
  $$
\end{itemize}

	Finally, as before we provide an inductive definition of the
    translation $\Tbis{-}$ which uses subtransformations $\Tpbis{-}$ in
    order to remember the centre of the ball and leads to the
    construction of $\phit{\phi_{qf}}'(x_1,\ldots,x_n)$. We only
    detail the case
    $$
    \Tpbis{\rels{j}{k}{y}{z}} =\psiun(y,z) \ou
    \psideux(y,z)
    $$
    as the other cases are identical as for the
    translation $\T{-}$ shown in the previous section. This leads to
    the following lemma (which is the pendant of Lemma
    \ref{lem:correct}).

    \begin{lemma} \label{lem:correct2}
	We have $\AA\models\phi_{qf}(\tuple{a})$ iff $\sem{\AA}'_{\tuple{a}}\models\phit{\phi_{qf}}'(\tuple{a})$.
\end{lemma}

As we had to characterise the well-formed $1$-data structure, a similar trick is necessary here. For  this matter, we use the following
formulas:
$$
  \begin{array}{ll}
\psitran &= \forall y \forall z.\Et_{p,q=1}^{n}\Et_{i,j,k,\ell=1}^D \Big(\udd{p}{i}{j}(y) \et \udd{p}{i}{\ell}(z) \et \udd{q}{k}{j}(y) \donc \udd{q}{k}{\ell}(z)\Big) \\
\psirefl(x_1,\ldots,x_n)  &=\Et_{p=1}^n\Et_{i=1}^D	 \udd{p}{i}{i}(x_p) \\

\psiwf(x_1,\ldots,x_n) &=\psitran \et \psirefl(x_1,\ldots,x_n) 
  \end{array}
  $$

Finally with the same reasoning as the one given in the previous
section, we can show that the formula $\phi=\exists x_1\ldots\exists
x_n.\linebreak[0]\phi_{qf}(x_1,\ldots,x_n)$ of $\eFO{D}{\Unary}{1}$ is satisfiable
iff the formula $\exists x_1\ldots\exists
x_n.\linebreak[0]\phit{\phi_{qf}}'(x_1,\ldots,x_n) \wedge \psiwf(x_1,\ldots,x_n) $
is satisfiable. Note that this latter formula  can be built in polynomial time from
$\phi$ and that it belongs to $\ndFO{0}{\Unaryp}$. Hence, thanks to
Theorem \ref{thm:0fo}, we obtain that $\nDataSat{\eFOr{1}}{D}$ is in
\textsc{NEXP}. The matching lower bound is as well obtained the same
way by reducing $\nDataSat{\ndFOr}{0}$ to $\nDataSat{\eFOr{1}}{D}$
showing that a formula $\phi$ in $\ndFO{0}{\Unary}$ is satisfiable
iff the formula $\exists
x.\locformr{\phi}{x}{1}$ in $\eFO{1}{\Unary}{D}$ is satisfiable.
  
\begin{theorem}
	For all $D \geq 1$, the problem $\nDataSat{\eFOr{1}}{D}$ is \textsc{NEXP}-complete.
\end{theorem}

%% file: undecidability.tex
We show here $\nDataSat{\eFOr{3}}{2}$ and  $\nDataSat{\eFOr{2}}{3}$
are undecidable. To obtain this we provide reductions from
$\nDataSat{\ndFOr}{2}$ and we use the fact that any
2-data structure can be interpreted as a radius-3-ball of a 2-data structure or
respectively  as a radius-2-ball of a 3-data structure.

\subsection{Radius 3 and two data values}

In order to reduce $\nDataSat{\ndFOr}{2}$ to
$\nDataSat{\eFOr{3}}{2}$, we show that 
we can transform slightly any $2$-data structure $\AA$ into an other
2-data structure $\AAge$ such that $\AAge$ corresponds to the
radius-3-ball of any element of $\AAge$  and this transformation has
some kind of inverse. Furthermore, given a formula $\phi \in
\ndFO{2}{\Unary}$, we transform it into a formula $T(\phi)$ in
$\eFO{2}{\Unary'}{3}$ such that $\AA$ satisfies $\phi$ iff $\AAge$
satisfies $T(\phi)$ . What follows is the formalisation of this reasoning.

Let  $\AA=(A,(\uP{\unary})_{\unary},\funi,\funo)$ be a $2$-data
structure in $\nData{2}{\Unary}$ and  $\uge$ be a fresh unary
predicate not in  $\Unary$. From $\AA$ we build the following $2$-data structure
$\AAge=(A',(\uP{\unary}')_{\unary},\funi',\funo')\in\nData{2}{\Unary\cup\{\uge\}}$
such that:
\begin{itemize}
	\item $A' = A \uplus \Values{\AA}\times\Values\AA$, 
	\item for $i\in\{1,2\}$ and $a\in A$, $f_i'(a)=f_i(a)$  and for $(d_1,d_2)\in  \Values{\AA}\times\Values\AA$, $f_i((d_1,d_2))=d_i$,
	\item for $\unary\in\Unary$, $\uP{\unary}'=\uP{\unary}$,
	\item $\uP{\uge}=\Values{\AA}\times\Values\AA$.
    \end{itemize}
Hence to build $\AAge$ from $\AA$ we have added to the elements of
$\AA$ all pairs of data presented in $\AA$ and in order to recognise
these new elements in the structure we use the new unary predicate
$\uge$. We add these extra elements to ensure that all the elements of
the structure are located in   the
radius-3-ball of any element of $\AAge$.
We have then the following property.

\begin{lemma}\label{lem:ge-has-radius-3}
 $\vprojr{\AAge}{a}{3}=\AAge$ for all $a \in A'$.
\end{lemma}
\begin{proof}
	Let $b\in A'$ and $i,j \in \{1,2\}$. We show that
    $\distaa{(a,i)}{(b,j)}{\AAge}\leq 3$. i.e. that there is a path of length at most 3 from $(a,i)$ to $(b,j)$ in the data graph $\gaifmanish{\AAge}$.
	By construction of $\AAge$, there is an element $c\in A'$ such that $f_1(c)=f_i(a)$ and $f_2(c)=f_j(b)$.
	So we have the path $(a,i),(c,1),(c,2),(b,j)$ of length at most 3 from $(a,i)$ to $(b,j)$ in $\gaifmanish{\AAge}$.
\end{proof}

Conversely, to $\AA=(A,(\uP{\unary})_{\unary},\funi,\funo)\in\nData{2}{\Unary\cup\{\uge\}}$, we associate $\AAminusge=(A',(\uP{\unary}')_{\unary},\funi',\funo')\in\nData{2}{\Unary}$ where:
\begin{itemize}
	\item $A' = A \setminus \uP{\uge}$, 
	\item for $i\in\{1,2\}$ and $a\in A'$, $f_i'(a)=f_i(a)$,
	\item for $\unary\in\Unary$, $\uP{\unary}'=\uP{\unary}'\setminus \uP{\uge}$.
\end{itemize}

Finally we inductively translate any formula $\phi\in\ndFO{2}{\Unary}$ into $T(\phi)\in\ndFO{2}{\Unary\cup\{\uge\}}$ by making it quantify over elements not labeled with $\uge$:  $T(\unary(x)) = \unary(x)$, $T(\rels{i}{j}{x}{y})=\rels{i}{j}{x}{y}$, $T( x=y )= (x=y) $, $T(\exists x.\phi)=\exists x. \neg \uge(x) \wedge T(\phi)$, $T( \vp \vee \vp')=T(\vp) \vee T(\vp')$ and $T(\neg \vp)=\neg T(\vp)$. 

\begin{lemma}\label{lem:ge-vs-without}
	Let $\phi$ be a sentence in $\ndFO{2}{\Unary}$,
    $\AA\in\nData{2}{\Unary}$ and $\BB \in
    \nData{2}{\Unary\cup\{\uge\}}$. The two following properties hold:
 \begin{itemize}
	 \item $\AA\models\phi$ iff $\AAge\models T(\phi)$
  \item $\minusge{\BB} \models\phi$ iff $\BB\models T(\phi)$.
 \end{itemize}
\end{lemma}

\begin{proof}
	As for any $\AA\in\nData{2}{\Unary}$ we have $\minusge{(\addge{\AA})} = \AA$, it is sufficient to prove the second point.
	We reason by induction on $\phi$.
	Let $\AA=(A,(\uP{\unary})_{\unary},\funi,\funo)\in\nData{2}{\Unary\cup\{\uge\}}$ and let $\AAminusge=(A',(\uP{\unary}')_{\unary},\funi',\funo')\in\nData{2}{\Unary}$. 
	The inductive hypothesis is that for any formula $\phi\in\ndFO{2}{\Unary}$ (closed or not) and any context interpretation function $I: \Var \to A'$ we have $\AAminusge \models_I \phi \text{ iff } \AA \models_I T(\phi)$.
	Note that the inductive hypothesis is well founded in the sense that the interpretation $I$ always maps variables to elements of the structures.
	
	We prove two cases: when $\phi$ is a unary predicate and when
    $\phi$ starts by an existential quantification, the other cases
    being similar. First, assume that $\phi = \unary(x)$ where $\unary\in\Unary$.
	$\AAminusge \models_I \unary(x)$ holds iff $I(x)\in\uP{\unary}'$. 
	As $I(x)\in A\setminus \uP{\uge}$, we have $I(x)\in\uP{\unary}'$ iff $I(x)\in\uP{\unary}$, which is equivalent to  $\AA \models_I T(\unary(x))$ .
	Second assume $\phi = \exists x.\phi'$.
	Suppose that $\AAminusge \models_I \exists x.\phi'$.
	Thus, there is a $a\in A'$ such that $\AAminusge \models_\Intrepl{x}{a} \phi'$.
	By inductive hypothesis, we have $\AA\models_\Intrepl{x}{a} T(\phi')$.
	As $a\in A' = A\setminus \uP{\uge}$, we have $\AA\models_\Intrepl{x}{a} \neg\uge(x)$, so $\AA\models_I \exists x. \neg\uge(x)\et T(\phi')$ as desired.
	Conversely, suppose that $\AA \models_I T(\exists x.\phi') $.
	It means that there is a $a\in A$ such that $\AA \models_\Intrepl{x}{a}\neg\uge(x)\et T(\phi')$.
	So we have that $a\in A'=A\setminus \uP{\uge}$, which means that $\Intrepl{x}{a}$ takes values in $A$ and we can apply the inductive hypothesis to get that 
	$\AAminusge \models_\Intrepl{x}{a} \phi'$.
	So we have $\AAminusge \models_I \exists x.\phi'$.
\end{proof}

From Theorem \ref{thm:undec-general}, we know that
$\nDataSat{\ndFOr}{2}$ is undecidable. From  a closed formula
$\phi\in\ndFO{2}{\Unary}$, we build the formula $\exists
x.\locformr{T(\phi)}{x}{3}\in\eFO{2}{\Unary\cup\{\uge\}}{3}$. Now if
$\phi$ is satisfiable, it means that there exists $\AA\in
\nData{2}{\Unary}$  such that $\AA\models\phi$. By Lemma
\ref{lem:ge-vs-without}, $\AAge\models T(\phi)$. Let $a$ be an element
of $\AA$, then thanks to Lemma \ref{lem:ge-has-radius-3}, we have $\vprojr{\AAge}{a}{3}\models T(\phi)$.
 Finally by definition of our logic, $\AAge\models\exists x.\locformr{T(\phi)}{x}{3}$.
 So $\exists x.\locformr{T(\phi}{x}{3}$ is satisfiable. Now assume
 that $\exists x.\locformr{T(\phi)}{x}{3}$ is satisfiable. So there
 exist $\AA \in \nData{2}{\Unary\cup\{\uge\}}$ and an element $a$ of
 $\AA$ such that $\vprojr{\AA}{a}{3}\models T(\phi)$.
 Using  Lemma \ref{lem:ge-vs-without}, we obtain
 $(\vprojr{\AA}{a}{3})_{\setminus\uge}\models\phi$. Hence $\phi$ is
 satisfiable. This shows that we can reduce $\nDataSat{\ndFOr}{2}$ to $\nDataSat{\eFOr{3}}{2}$ .

\begin{theorem}\label{thm:undec-existential-r3}
	The problem $\nDataSat{\eFOr{3}}{2}$ is undecidable.
\end{theorem}

\subsection{Radius 2 and three data values}

We provide here a reduction from $\nDataSat{\ndFOr}{2}$ to
$\nDataSat{\eFOr{2}}{3}$. The idea is similar to the one used in the
proof of Lemma \ref{lem:hardness-radius2-2} to show that
$\nDataSat{\eFOr{2}}{2}$ is \textsc{N2EXP}-hard by reducing
$\nDataSat{\ndFOr}{1}$. Indeed we have the following Lemma.

\begin{lemma} 
 Let $\phi$ be
a formula in $\ndFO{2}{\Unary}$. There exists $\AA \in \nData{2}{\Unary}$ such that $\AA \models \phi$
 if and only if there exists $\BB \in  \nData{3}{\Unary}$ such that
 $\BB \models \exists
x.\locformr{\phi}{x}{2}$.
\end{lemma}

\begin{proof}
 Assume that there exists  $\AA=(A,(P_{\unary})_{\unary \in
   \Unary},\f{1},\f{2})$ in $\nData{2}{\Unary}$ such that  $\AA \models
 \phi$.Consider the $3$-data structure $\BB=(A,(P_{\unary})_{\unary \in
   \Unary},\f{1},\f{2},\f{3})$ such that $\f{3}(a)=0$ for all $a\in
 A$. Let $a \in A$. It is clear that we have $\vprojr{\BB}{a}{2}=\BB$
 and that $\vprojr{\BB}{a}{2} \models \phi$ (because $\AA \models
 \phi$ and $\phi$ never mentions the third values of the elements
 since it is a formula in $\ndFO{1}{\Unary}$). Consequently $\BB
 \models \exists
 x.\locformr{\phi}{x}{2}$.

 Assume now that there exists $\BB=(A,(P_{\unary})_{\unary \in
   \Unary},\f{1},\f{2},\f{3})$ in $ \nData{3}{\Unary}$ such that  $\BB \models \exists
x.\locformr{\phi}{x}{2}$. Hence there exists $a \in A$ such that
$\vprojr{\BB}{a}{2} \models \phi$, but then by forgetting the third
value in $\vprojr{\BB}{a}{2}$ we obtain a model in $\nData{3}{\Unary}$
which satisfies $\phi$.
\end{proof}

Using Theorem \ref{thm:undec-general}, we obtain the following result.
\begin{theorem}\label{thm:undec-existential-r2}
	The problem $\nDataSat{\eFOr{2}}{3}$ is undecidable.
\end{theorem}